\documentclass[12pt]{amsart}
\usepackage{amsmath,amsthm, graphicx, amsfonts}
\usepackage{color}

\def\be{\begin{equation}}
\def\ee{\end{equation}}
\def\ba{\begin{array}}
\def\ea{\end{array}}
\def\bc{\begin{center}}
\def\ec{\end{center}}

\def\Rd{{\RR}^d}

\def\ZZ{\rm {{\rm Z}\kern-.48em{\rm Z}}}
\def\RR{\rm \hbox{I\kern-.2em\hbox{R}}}
\def\CC{\rm \hbox{C\kern -.5em{\raise .32ex
\hbox{$\scriptscriptstyle |$}}\kern - .22em{\raise .6ex
\hbox{$\scriptscriptstyle |$}}\kern
.4em}}
\def\CC{ {\mathbf C} }
\def\RR{ {\mathbf R} }
\def\ZZ{ {\mathbf Z} }

\newtheorem{Tm}{Theorem}[section]
\newtheorem{Df}[Tm]{Definition}
\newtheorem{Cr}[Tm]{Corollary}

\newtheorem{Lm}[Tm]{Lemma}
\newtheorem{re}[Tm]{Remark} 

\numberwithin{equation}{section}

\begin{document}
\bibliographystyle{plain}

\title[Galerkin Reconstruction in reproducing kernel spaces] {Sampling and  Galerkin reconstruction  in reproducing kernel spaces
}

\author{Cheng Cheng, Yingchun Jiang,  and Qiyu  Sun}
\address{Cheng: Department of Mathematics, University of Central Florida, Orlando, Florida 32816, USA}
\email{cheng.cheng@knights.ucf.edu}

\address{Jiang: School of Mathematics and Computational Science, Guilin University of
 Electronic Technology, Guilin, Guangxi 541004, China.
}
\email{guilinjiang@126.com}

\address{Sun: Department of Mathematics, University of Central Florida, Orlando, Florida 32816, USA}
\email{qiyu.sun@ucf.edu}


\begin{abstract}
In this paper, we consider sampling in a reproducing kernel subspace of $L^p$.
 We introduce a pre-reconstruction operator associated with a sampling scheme and propose a Galerkin reconstruction
  in general Banach space setting. We show that the proposed Galerkin method  provides a quasi-optimal approximation, and the corresponding
Galerkin equations  could be solved by an iterative approximation-projection algorithm.
  We also  present  detailed analysis and numerical simulations of the Galerkin method for reconstructing signals with finite rate of innovation.
\end{abstract}



 \maketitle

\section{Introduction}

The celebrated Whittaker-Shannon-Kotelnikov's sampling theorem states that a
 bandlimited signal can be recovered from its samples taken at a rate greater than twice the bandwidth \cite{sire49, whittaker35}.
 In last two decades, that paradigm  has been extended
 to represent signals in a shift-invariant space \cite{aksiam2001, astjfaa2005, unser2000},
 signals with finite rate of innovation \cite{vetterli07, mishalieldar2011, panbludragotti14, sunaicm08, sunsiam06, vetterli02}, and signals in a reproducing kernel space \cite{christensen12, garcia13, hns09,  ns2010, nw91}.

\smallskip

In this paper, we consider signals living in a  reproducing kernel space (RKS) of the form
\begin{equation} \label{rks.def}
V_{K,p}:= \big\{T_0f:  \ f \in L^p\big\}=
\{f\in L^p:  \ T_0f=f\},\ 1\le p\le \infty,\end{equation}
where $T_0$ is an idempotent integral operator
 with kernel $K$,
 \begin{equation}\label{T0.def}
T_0f(x) := \int_{\RR^d}K(x,y)f(y)dy , \ \ f \in L^p.
\end{equation}
 The RKS  has rich geometric structure, lots of flexibility and technical suitability for sampling.
  It has been used for modeling  bandlimited signals, wavelet (spline) signals,
    and signals with finite rate of innovation
    \cite{aksiam2001, ns2010, nw91, sunaicm08, unser2000}.

\medskip

Take a (finite) sampling set $\Gamma$ and consider the sampling scheme
 \begin{equation*} 
 f\longmapsto \{f(\gamma_n), \gamma_n\in \Gamma\}, \ \ f\in V_{K,p}.\end{equation*}
 We are interested in
 finding a quasi-optimal linear approximation $Rf$, depending completely on the sampling data,  in a reconstruction space $U$ for a signal $f\in V_{K,p}$,
 $$\|Rf-f\|_p\le C \inf_{h\in U} \|f-h\|_p,\ \ f\in V_{K,p}.$$

\smallskip
In this paper, 
we focus on {\em  pre-reconstruction
 operators}
\begin{equation}\label{rkssampling.def}
S_{\Gamma, \delta} f(x):=\sum_{\gamma_n\in \Gamma} |I_n| f(\gamma_n) K(x, \gamma_n),\ \
f\in V_{K,p},
\end{equation}
where $\delta>0$ and
$\{I_n\subset B(\gamma_n, \delta): \ \gamma_n\in \Gamma\}$ is  a disjoint covering of
\begin{equation*} 
 B(\Gamma, \delta):=\cup_{\gamma\in \Gamma}  B(\gamma, \delta)
=\cup_{\gamma\in \Gamma} \{x: \  |x-\gamma|\le \delta\}.
\end{equation*}
Our crucial observation is that    $S_{\Gamma, \delta}f(x)$
is a good approximation to   $f(x)$
  when $\delta$ is  sufficiently small and
$x\in B(\Gamma, \delta)$ is far away from the  complement of
$B(\Gamma, \delta)$, see Figure \ref{prereconstruction.fig} in Section \ref{simulation.section}.


\smallskip

Associated with the pre-reconstruction operator $S_{\Gamma, \delta}$, we introduce  the  Garlekin method
\begin{equation}\label{galerkinrks.def}
\langle S_{\Gamma, \delta} Rf, g\rangle=\langle S_{\Gamma, \delta} f, g\rangle,\ \ g\in \tilde U\subset L^{p/(p-1)}\end{equation}
to  define a quasi-optimal linear approximation $Rf$ in the reconstruction space $U$, where
$\langle \cdot, \cdot\rangle$ is the standard dual product between  $L^p$
and  $L^{p/(p-1)}$. We recognize that
 the  Galerkin equation \eqref{galerkinrks.def} could be solved by certain iterative approximation-projection algorithm:
\begin{equation}\label{iaprks.def}
 g_0\in U\ \ {\rm and} \ \  g_{m+1}=g_m-P_{U, \tilde U} S_{\Gamma, \delta}g_m+g_0,\  m\ge 0,
\end{equation}
where  $P_{U, \tilde U}$ is an oblique projection  for
the trial-test space pair $(U, \tilde U)$, c.f.
 \cite{afpams98, astca04, fgsiam92,  ns2010, sunwzhou2002}.

\medskip
This paper is organized as follows.
In Section \ref{banach.section}, we introduce  the concept of admissibility of pre-reconstruction operators in Banach space setting. We show  that
(sub-)Galerkin reconstruction provides a quasi-optimal approximation (Theorem  \ref{nsqs.thm}), and such (sub-)Galerkin reconstruction exists
 whenever the trial and test spaces are finite-dimensional 
(Theorem \ref{Rexistence.thm}, Corollaries \ref{Rexistence.cor} and \ref{hilbertreconstruction.cor}).
In Section \ref{rks.section}, we discuss  admissibility of the pre-reconstruction operator $S_{\Gamma, \delta}$ in \eqref{rkssampling.def} (Theorem \ref{rksadmissible.thm}).
In that section, we also propose to use the  iterative approximation-projection algorithm \eqref{iaprks.def}
 to solve the Galerkin equation \eqref{galerkinrks.def} (Theorem \ref{galerkinrks.thm} and Lemma \ref{galerkinrks.lem}).
 Lots of signals with finite rate of innovation live in some
 reproducing kernel spaces of the form \eqref{rks.def}.
 In Section \ref{fri.section}, we provide  detailed analysis for pre-reconstruction operators, and we obtain matrix formulation of Galerkin reconstructions
  for signals with finite rate of innovation.
 In last section, we present some numerical simulations to demonstrate our  Galerkin  method.

\section{Sub-Galerkin reconstruction in Banach spaces}
\label{banach.section}

In this section, we consider numerical stability and quasi-optimality  of a (sub-)Galerkin reconstruction in  Banach space setting.

\smallskip

Denote by $\langle \cdot, \cdot\rangle$
 the action between elements in  a Banach space $B$ and its dual space $B^*$.
First we introduce  admissibility of operators for the trial-test space pair.

\begin{Df}\label{admissible.def}
Let $(U, V, B)$ be a triple of Banach spaces with $U\subset V\subset B$, and let $\tilde U\subset B^*$.
We say that a bounded linear operator $S:V\to V$ is  {\rm admissible} for the
trial-test space pair $(U, \tilde U)$
if there exist positive constants $D_1$ and $D_2$ such that
\begin{equation}\label{ad.def1}
\sup_{g\in \tilde U, \|g\|\le 1}|\langle Sf, \ g\rangle|\ge D_1\|f\|\ \ {\rm  for\ all} \ \ f\in U,
\end{equation}
and
\begin{equation}\label{ad.def2}
\sup_{g\in \tilde U, \|g\|\le 1}|\langle Sf, \ g\rangle|\le D_2 \|f\|\ \ {\rm  for\ all} \  \ f\in V.
\end{equation}


\end{Df}

An admissible operator $S$ for the trial-test space pair $(U, \tilde U)$
 is bounded below on $U$,
$$\|Sf\|\ge D_1\|f\|, \ \  f\in U.$$
 The performance of our proposed (sub-)Galerkin reconstruction depends on  the test space $\tilde U$, particularly
 on  the ratio between  bounds $D_1$
 and  $D_2$
in  \eqref{ad.def1} and \eqref{ad.def2}, see Theorem \ref{nsqs.thm}. In our model for sampling, $S$ is the pre-reconstruction operator $S_{\Gamma, \delta}$ in \eqref{rkssampling.def},
and the triple of Banach spaces contains  the reconstruction space $U$, the reproducing kernel space $V_{K,p}$ and the space $L^p$.

\smallskip

 Next we introduce a  general notion of Galerkin reconstruction.

\begin{Df} Let $S: V\to V$ be a bounded  linear operator,
and  $(U, \tilde U)$ be a trial-test space pair.
 We say that  a linear operator $R:V\to U$ is a {\em Galerkin reconstruction}  if
\begin{equation} \label{r.def0} Rh=h,  \ h\in U\end{equation}
and
\begin{equation} \label{r.def1} \langle SRf, g\rangle= \langle Sf, g\rangle, \ f\in V\ {\rm and} \ g\in \tilde U;\end{equation}
and  a {\em sub-Galerkin reconstruction} if
 \eqref{r.def0}  holds and
\begin{equation}\label{r.def2}
\sup_{g\in \tilde U, \|g\|\le 1} |\langle SRf, \ g\rangle|\le  D_3
\sup_{g\in \tilde U, \|g\|\le 1} |\langle Sf, \ g\rangle|,  \ f\in V,\end{equation}
for some   $D_3>0$.
\end{Df}

In the following theorem,  we establish  numerical stability and quasi-optimality
of (sub-)Galerkin reconstructions associated with  admissible operators.

\begin{Tm} \label{nsqs.thm}
Let $V, U, \tilde U$ be as in Definition \ref{admissible.def}, and
 $S$ be admissible for the pair $(U, \tilde U)$ with bounds $D_1$ and $D_2$.  If  $R: V\to U$
 is a sub-Galerkin reconstruction with bound $D_3$, then
\begin{itemize}
\item [{(i)}] $R$ is  numerically stable,
 \begin{equation*} 
\|Rf\|\le \frac{D_2D_3}{D_1} \|f\|,  \ f\in V. 
\end{equation*}
\item [{(ii)}] $R$ is quasi-optimal,
\begin{equation*} 
\|Rf-f\|\le \frac{D_1+D_2 D_3}{D_1}
\inf_{h\in U} \|f-h\|, \ f\in V. 
\end{equation*}
\end{itemize}
\end{Tm}

\begin{proof}  (i)\ \   For  $f\in V$, we obtain from \eqref{ad.def1}, \eqref{ad.def2} and \eqref{r.def2} that
$$D_{1}\|Rf\|\leq \sup \limits_{g\in \tilde U,\|g\|\le 1}|\langle SRf, \ g\rangle|
\leq D_{3}
 \sup \limits_{g\in \tilde U,\|g\| \le 1}|\langle Sf, \ g\rangle| \leq D_2D_3\|f\|.$$
This proves numerical stability of the reconstruction operator $R$.


(ii)\ \
For  $f\in V$ and $h\in U$,
\begin{eqnarray*}
\|f-Rf\|  & \leq &  \|f-h\|+\|h-Rf\|\\
&= & \|f-h\| +\|R(f-h)\| \leq \frac{D_1+D_2D_3}{D_1}\|f-h\|,\end{eqnarray*}
where we have used the facts that $R$ is a sub-Galerkin reconstruction and  has numerical stability. 
Then   quasi-optimality
of the reconstruction operator $R$ holds by taking infinimum over  $h\in U$.
\end{proof}

 By Theorem \ref{nsqs.thm},
the existence of a quasi-optimal approximation reduces to
finding a sub-Galerkin reconstruction. Now we show that such
a sub-Galerkin reconstruction always exists when $U$ and $\tilde U$  are finite-dimensional.

\begin{Tm}\label{Rexistence.thm} Let $V, U, \tilde U$ be as in Definition \ref{admissible.def}, and
 $S$ be admissible for the pair $(U, \tilde U)$.
  If
 $U$ and $\tilde U$ are finite-dimensional, then
 there is a sub-Galerkin  reconstruction. 
\end{Tm}

\begin{proof}
 Let $\{f_{i}\}_{i=1}^{m}$ and $\{g_{i}\}_{i=1}^{n}$ be bases of  $U$ and $\widetilde{U}$ respectively.
 By the admissibility of $S$,  we may assume that
  $B:=(\langle Sf_{i},g_{j}\rangle)_{1\leq i, j\leq m}$ is nonsingular.
 Write $B^{-1}=(b_{ij})$ and define linear operator  $R$ by
 \begin{equation*} 
 Rf:=\sum_{i,j=1}^m \langle Sf, g_i\rangle b_{ij} f_j, \ f\in V.
 \end{equation*}
  Obviously, $R$ satisfies \eqref{r.def0}. Now it remains to
  show that $R$ satisfies
  \eqref{r.def2}.




   Let $\tilde U_\ast$ be  the space 
  spanned by $\{g_j\}_{j=1}^m$.
  One may verify that   $Rf$ solves  Galerkin equations
 \begin{equation} \label{Rexistence.thm.pf.eq5}
 \langle SRf, g\rangle
 =\langle Sf, g\rangle, \ g\in \tilde U_\ast
 \end{equation}
 for any $f\in V$,  and 
   \begin{equation} \label{Rexistence.thm.pf.eq2}
C_0\|h\|\le \sup \limits_{g\in\widetilde{U}_{\ast},\|g\|\le 1}|\langle Sh, \ g\rangle|,\  h\in U
 \end{equation}
 for some positive constant $C_0$.
Therefore 
 \begin{eqnarray*}
  \sup \limits_{g\in\widetilde{U},\|g\|\le 1}|\langle SRf, \ g\rangle|
  &\leq & D_{2}\|Rf\|\\
& \le &   D_2 (C_0)^{-1}
 \sup \limits_{g\in\widetilde{U}_{\ast},\|g\|\le1}|\langle SRf, \ g\rangle|\\
 &  = & D_2 (C_0)^{-1}
 \sup \limits_{g\in\widetilde{U}_{\ast},\|g\|\le1}|\langle Sf, \ g\rangle|\\
 & \leq &  D_2 (C_0)^{-1} \sup \limits_{g\in\widetilde{U},\|g\|\le1}|\langle Sf, \ g\rangle|, \ f\in V,
  \end{eqnarray*}
 by \eqref{Rexistence.thm.pf.eq5}, \eqref{Rexistence.thm.pf.eq2} and the admissibility of  $S$.
\end{proof}


For the case that  $U$ and $\tilde U$ have the same dimension,
 we have

\begin{Cr}\label{Rexistence.cor}
Let $V, U, \tilde U$ be as in Definition \ref{admissible.def}, and
 $S$ be admissible for the pair $(U, \tilde U)$.
   If dimensions of
 $U$ and $\tilde U$ are the same,
 then for  $f\in V$,
 the unique solution of Galerkin equations
 \begin{equation}\label{Galerkin.eq}
 \langle SRf, g\rangle = \langle Sf , g\rangle, \  g \in \tilde{U},
 \end{equation}
defines
a Galerkin  reconstruction. 
\end{Cr}


In Hilbert space setting, 
 we  can establish the following  result for least squares solutions.

\begin{Cr}\label{hilbertreconstruction.cor}
Let $V$ be a Hilbert space, $U$ and $\tilde U$ be linear subspaces of $V$, and let
 $S$ be admissible for the pair $(U, \tilde U)$. If $U$  and $\tilde U$ are finite-dimensional, then
 the least squares solution of  Galerkin equations \eqref{Galerkin.eq},
  \begin{equation*} 
  Rf:={\rm argmin}_{h\in U}\sup_{g\in \tilde U, \|g\|\le 1} | \langle S(h-f), g\rangle |, \ \ f\in V,
 \end{equation*}
defines a sub-Galerkin reconstruction with bound $D_3\le 1$. 
\end{Cr}

The above conclusion on least squares solutions 
with $\tilde U=U$ has been established by
Adcock, Gataric and Hansen  for non-uniform sampling \cite{agh2014, agh2014b}.

\section{Sampling and Reconstruction in 
$V_{K,p}$}
\label{rks.section}

To consider sampling and reconstruction in $V_{K,p}$, we {\em always} assume that
 the kernel $K$ of the space $V_{K,p}$ in \eqref{rks.def} satisfies
\begin{equation}
\label{kernel.req1}
 \|K\|_{\mathcal W} := \max \Big\{\sup_{x\in \RR^d} \|K(x,\cdot)\|_{1} ,
 \sup_{y\in \RR^d} \|K(\cdot,y)\|_1\Big\} < \infty
 \end{equation}
and
\begin{equation}\label{kernel.req2}
 \lim \limits_{\delta \rightarrow 0}\|\omega_\delta(K)\|_{\mathcal W} = 0, 
 \end{equation}
 where
$$\omega_\delta(K)(x,y):= \sup \limits_{|x^\prime|,|y^\prime|\le \delta}|K(x+x^\prime,y+y^\prime)-K(x,y)|.$$
Under the above hypothesis, the integral operator $T_0$ in \eqref{T0.def}
is a bounded operator on $L^p$,
$$\|T_0f\|_p\le \|K\|_{\mathcal W} \|f\|_p, \ \ f\in L^p.$$
More importantly, its range space
 $V_{K,p}$
is a reproducing kernel space
\cite{ns2010}. In this section, we consider admissibility of the pre-reconstruction operator $S_{\Gamma, \delta}$ in \eqref{rkssampling.def}
 and the unique
 Galerkin reconstruction associated with it.

 \smallskip

\subsection{Admissibility, stability and samplability}
  To discuss the admissibility, we introduce the {\em residue} $E(U, F)$
 of  signals in a linear space $U\subset L^p$
outside a measurable set $F$,
\begin{equation*}\label{etf.def}
E(U, F):=
\sup_{0\ne f\in U}\frac{\|f\|_{L^p(\RR^d\backslash F)}}{\|f\|_p},
\end{equation*}
where $\|\cdot\|_{L^p(E)}$ is the $p$-norm on a measurable set $E$.
The reader may refer to \cite{agh2014, lakeybook, jaming2014} for  some applications of residues of bandlimited signals.

\begin{Tm}\label{rksadmissible.thm}
Let $V_{K,p}$ and  $S_{\Gamma, \delta}$  be as in \eqref{rks.def} and \eqref{rkssampling.def} respectively.
Assume that $U\subset V_{K,p}$ and $\tilde U\subset L^{p/(p-1)}$.
If
\begin{equation}\label{rksadmissible.thm.eq1}
\sup_{g\in \tilde U, \|g\|_{p/(p-1)}\le 1}
|\langle f, g\rangle| 
\ge  D_4
\|f\|_p, \   f\in U \end{equation}
for some  constant $D_4$ satisfying
\begin{equation}\label{rksadmissible.thm.eq2}
r_0:= D_4^{-1} \big(E(U, B(\Gamma, \delta)) \|K\|_{\mathcal W}
+ \|\omega_{\delta}(K)\|_{\mathcal W}
\big(1+ \|K\|_{\mathcal W}+
\|\omega_{\delta}(K)\|_{\mathcal W}\big)\big)<1,\end{equation}
then  $S_{\Gamma, \delta}$  is admissible for the pair $(U, \tilde U)$.
\end{Tm}

%
%
%
%

 Given a sampling set $\Gamma$, we say that the sampling scheme 
 \begin{equation}\label{samplingscheme.def22}
 U\ni f\longmapsto 
 \{f(\gamma_n), \gamma_n\in \Gamma\}\end{equation}
   has {\em weighted $\ell^p$-stability on $U$}
 if  there exist  positive constants
 $C_1, C_2$ and $\delta$ such that
 \begin{equation*}
 C_1\|f\|_p\le  \Big(\sum_{\gamma_n\in \Gamma} |I_n| |f(\gamma_n)|^p\Big)^{1/p}\le C_2\|f\|_p, \ f\in U,
 \end{equation*}
 if $1\le p<\infty$, and
 \begin{equation*}
 C_1\|f\|_\infty\le  \sup_{\gamma_n\in \Gamma}  |f(\gamma_n)|\le C_2\|f\|_\infty, \ f\in U,
 \end{equation*}
 if $p=\infty$,
 where  $\{I_n\subset B(\gamma_n, \delta), \gamma_n\in \Gamma\}$ is  a disjoint covering of
the $\delta$-neighborhood $B(\Gamma, \delta)$ of the sampling set $\Gamma$.
Weighted stability of a sampling scheme  implies its  unique  determination. It is an important concept for robust signal reconstruction, see \cite{aksiam2001, astca04,bg2013,eldar05, ns2010,sunsiam06,sunxian2014, sunwzhou2002, unser2000} and  references here.
The following result  connects the  weighted $\ell^p$-stability of
a sampling scheme  with the admissibility of a pre-reconstruction operator.

\begin{Tm}\label{stabilityrks.tm}
Let $V_{K,p}$ and $S_{\Gamma, \delta}$ be as in \eqref{rks.def} and
\eqref{rkssampling.def} respectively.
Assume that $U\subset V_{K,p}$ and $\tilde U\subset L^{p/(p-1)}$.
If $S_{\Gamma, \delta}$ is admissible for the pair $(U, \tilde U)$, then
the sampling scheme \eqref{samplingscheme.def22} on $\Gamma$
  has weighted $\ell^p$-stability on $U$.
\end{Tm}

By the regularity assumption \eqref{kernel.req2} on the reproducing kernel $K$,
the second requirement \eqref{rksadmissible.thm.eq2} in Theorem
\ref{rksadmissible.thm} is satisfied if $\delta$ is sufficiently small and 
$B(\Gamma, \delta)$ 
 is the whole Euclidean space
$\RR^d$. For the case that $B(\Gamma, \delta)$  contains an open domain $F_0$ but not necessarily the whole space $\RR^d$,  we obtain the following samplability result from
Theorems \ref{rksadmissible.thm} and \ref{stabilityrks.tm}.

%

\begin{Cr} Let $U\subset V_{K,p}$ and $D_4$
 be as in  Theorem \ref{rksadmissible.thm}. Assume that $F_0$ is an open domain satisfying
$E(U, F_0)\|K\|_{\mathcal W}<D_4$. If
$\Gamma$ is a sampling set with  $B(\Gamma, \delta)\supset F_0$ for some sufficiently small  $\delta>0$,
then
 signals in $U$
are uniquely determined by their samples taken on $\Gamma$.
 \end{Cr}

The samplability of various signals 
is well-studied, 
see, e.g.,
 \cite{agh2014b, fgsiam92, gmc92}  for band-limited signals,
 \cite{aksiam2001, unser2000}
 for signals in a shift-invariant space,
 \cite{sunaicm08, sunsiam06} for signals with finite rate of innovation, and
 \cite{hns09, ns2010} for signals in  a reproducing kernel 
  space.

\bigskip

To prove Theorem \ref{rksadmissible.thm},
%
 we need the following lemma.

\begin{Lm}\label{Sfupperbound.lem} Let $V_{K,p}$ and $S_{\Gamma, \delta}$ be as in \eqref{rks.def} and
\eqref{rkssampling.def} respectively. Then
\begin{equation*}\label{Sfupperbound.lem.eq1}
\|S_{\Gamma, \delta}f\|_p\le\big (\|K\|_{\mathcal W}+ \|\omega_{\delta}(K)\|_{\mathcal W}\big)
\big(1+ \|\omega_{\delta}(K)\|_{\mathcal W}\big) \|f\|_p, \  f\in V_{K,p}.
\end{equation*}
\end{Lm}

\begin{proof}
Let $\{I_n\}$ be the disjoint covering of
 $B(\Gamma, \delta)$  
 in \eqref{rkssampling.def}.
For $f\in V_{K,p}$, write
\begin{eqnarray}\label{rksadmissible.lem.pf.eq1}
S_{\Gamma, \delta}f(x) & = &  
\sum_{n} \int_{I_n}
\int_{\RR^d} K(x, \gamma_n) K(\gamma_n,z)f(z) dzdy\nonumber\\
& = & \sum_{n} \int_{I_n} \int_{\RR^d} \Big\{
K(x,y) K(y,z) + (K(x, \gamma_n)-K(x,y)) \nonumber\\
& &\
\times K(y,z)
+K(x, y) (K(\gamma_n,z)-K(y,z)) \nonumber\\
 & &\ +(K(x, \gamma_n)-K(x,y))   (K(\gamma_n,z)-K(y,z))
\Big\} f(z)
dzdy\nonumber\\
&  =: &\uppercase\expandafter{\romannumeral1} + \uppercase\expandafter{\romannumeral2}+\uppercase\expandafter{\romannumeral3} + \uppercase\expandafter{\romannumeral4}.
\end{eqnarray}
Observe that
\begin{equation*}
\|\uppercase\expandafter{\romannumeral1}\|_p=
\Big\|\int_{B(\Gamma, \delta)} K(\cdot, y) f(y) dy\Big\|_p\le \|K\|_{\mathcal W}\|f\|_p,
\end{equation*}
\begin{equation*}
\|\uppercase\expandafter{\romannumeral2}\|_p\le
\Big\|\int_{\RR^d} \omega_{\delta}(K)(\cdot, y) |f(y)| dy\Big\|_p\le
\|\omega_{\delta}(K)\|_{\mathcal W}\|f\|_p,
\end{equation*}
\begin{eqnarray*}
\|\uppercase\expandafter{\romannumeral3}\|_p & \le &
\Big\|\int_{\RR^d} \int_{\RR^d}
|K(\cdot, y)| \omega_{\delta}(K)(y, z)  |f(z)| dz dy\Big\|_p\nonumber\\
& \le & \|K\|_{\mathcal W}
\|\omega_{\delta}(K)\|_{\mathcal W}
\|f\|_p,
\end{eqnarray*}
and
\begin{eqnarray*}
\|\uppercase\expandafter{\romannumeral4}\|_p & \le &
\Big\|\int_{\RR^d} \int_{\RR^d}
\omega_{\delta}(K)(\cdot, y) \omega_{\delta}(K)(y, z)  |f(z)| dz dy\Big\|_p\nonumber\\
& \le &
\|\omega_{\delta}(K)\|_{\mathcal W}^2 \|f\|_p.
\end{eqnarray*}
Combining the above four estimates with
\eqref{rksadmissible.lem.pf.eq1} completes the proof. 
\end{proof}

We finish this subsection with proofs of Theorems \ref{rksadmissible.thm} and \ref{stabilityrks.tm}.
\begin{proof}[Proof of Theorem \ref{rksadmissible.thm}]
The upper bound estimate
\eqref{ad.def2} for  the  operator $S_{\Gamma, \delta}$ 
 follows immediately from
Lemma \ref{Sfupperbound.lem}.

 Define 
 $$T^*_0 g(x):=\int_{\RR^d} K(y, x) g(y) dy,\  g\in L^{p/(p-1)}.$$
 For $f\in U$ and $g\in \tilde U\subset L^{p/(p-1)}$ with $\|g\|_{p/(p-1)}\le 1$, we obtain
  \begin{eqnarray} \label{rksadmissible.thm.pf.eq1}
 |\langle S_{\Gamma, \delta}f, g\rangle-\langle f, g\rangle|
& \le &
\Big|\int_{\RR^d\backslash B(\Gamma, \delta)}
 f(x) T_0^* g(x) dx\Big|\nonumber\\
 & &  + \Big|\sum_{n} \int_{I_n} f(\gamma_n)
  (T_0^\ast g)(\gamma_n)-f(x)(T_0^\ast g)(x)dx\Big|\nonumber \\
  & \le & \|K\|_{\mathcal W}
 \|f\|_{L^p(\RR^d\backslash B(\Gamma, \delta))} \nonumber\\
 & & +
\|\omega_{\delta}(K)\|_{\mathcal W}
\big(1+ \|K\|_{\mathcal W}+
\|\omega_{\delta}(K)\|_{\mathcal W}\big)
\|f\|_p,
\end{eqnarray}
where $\{I_n\}$  is the disjoint covering of
 $B(\Gamma, \delta)$
 in \eqref{rkssampling.def}.
This together with
\eqref{rksadmissible.thm.eq1}
and \eqref{rksadmissible.thm.eq2} proves
the lower bound estimate \eqref{ad.def1}
 for  the
 operator $S_{\Gamma, \delta}$. 
\end{proof}

\begin{proof}[Proof of Theorem \ref{stabilityrks.tm}]
Take $f\in V$. Following the argument used in Lemma \ref{Sfupperbound.lem},
we obtain
\begin{equation*}
\big(\|K\|_{\mathcal W}+\|\omega_{\delta}(K)\|_{\mathcal W}\big)^{-1}\|S_{\Gamma, \delta} f\|_p\le
\Big(\sum_{n} |I_n| |f(\omega_n)|^p\Big)^{1/p}
\le  \big(1+\|\omega_\delta(K)\|_{\mathcal W}\big)\|f\|_p
\end{equation*}
for $1\le p<\infty$ and
\begin{equation*}
\big(\|K\|_{\mathcal W}+\|\omega_{\delta}(K)\|_{\mathcal W}\big)^{-1}\|S_{\Gamma, \delta} f\|_\infty\le
\sup_{n} |f(\omega_n)|
\le \|f\|_\infty
\end{equation*}
for $p=\infty$.
The above two estimates together with admissibility of the operator $S_{\Gamma, \delta}$
 complete the proof.
\end{proof}

\subsection{Galerkin reconstruction}
To consider Galerkin reconstruction associated with the operator $S_{\Gamma, \delta}$
on the reproducing kernel space $V_{K,p}$, we introduce
the oblique projection for a pair $(U, \tilde U)$ of Banach spaces.

\begin{Df}
Given $U\subset V_{K,p}$ and $\tilde U\subset L^{p/(p-1)}$,  a bounded operator $P_{U, \tilde U}:V_{K,p}\to U$ is said to be an {\em oblique projection} for the pair $(U, \tilde U)$ if
\begin{equation}\label{projection200}
P_{U, \tilde U} h=h, \ h\in U,
\end{equation}
and
\begin{equation}\label{projection201}
\langle P_{U, \tilde U}f, g\rangle=\langle f, g\rangle, \ f\in V_{K,p}, g\in \tilde U.
\end{equation}
\end{Df}

In Hilbert space setting, 
an oblique projection $P_{U, \tilde U}$ exists when
cosine of the subspace angle between $U$ and $\tilde U^\perp$ is positive \cite{ahp2013,
bg2013, eldar05, tang1999}.
Following the argument used in Theorem \ref{Rexistence.thm}, we can show that
if   $U$ and $\tilde U$ have the same dimension and satisfy the first requirement \eqref{rksadmissible.thm.eq1} of Theorem \ref{rksadmissible.thm},
  then there is  an oblique projection  $P_{U, \tilde U}$ for the pair $(U, \tilde U)$.

 \begin{Tm}\label{galerkinrks.thm} Let $V_{K,p}$ and  $S_{\Gamma, \delta}$  be as in \eqref{rks.def} and \eqref{rkssampling.def} respectively.
 Assume that  $U\subset V_{K,p}$ and $\tilde U\subset L^{p/(p-1)}$ satisfy
\eqref{rksadmissible.thm.eq1} and \eqref{rksadmissible.thm.eq2}, and
 an oblique projection $P_{U, \tilde U}$
associated with the pair $(U, \tilde U)$ exists. Then
 Galerkin equations
\begin{equation}\label{galerkinrks.thm.eq00}
\langle S_{\Gamma, \delta} h, g\rangle= \langle S_{\Gamma, \delta} f, g\rangle, \ g\in \tilde U,
\end{equation}
have a unique solution  $h\in U$ for $f\in V_{K,p}$. Moreover, the  mapping $f\to h$
defines a Galerkin reconstruction.
\end{Tm}

To solve  Galerkin equations
\eqref{galerkinrks.thm.eq00}, we need  exponential convergence of
 the iterative approximation-projection algorithm
 \eqref{iaprks.def}.
%
The algorithm \eqref{iaprks.def} has been demonstrated to be efficient to  reconstruct
various signals. 
The reader may refer to \cite{fgsiam92, sunwzhou2002} for band-limited signals,
\cite{afpams98, astca04} for signals in a shift-invariant space,
and \cite{ns2010} for signals in a reproducing kernel space.

%

\begin{Lm}
\label{galerkinrks.lem} 
 Let $V_{K,p}, S_{\Gamma, \delta}, U, \tilde U$ and $P_{U, \tilde U}$ be as in Theorem \ref{galerkinrks.thm},
 and let $r_0\in (0, 1)$ be as in \eqref{rksadmissible.thm.eq2}.
 Then for any $g_0\in U$,
the sequence $g_m, m\ge 0$, in the iterative algorithm
\eqref{iaprks.def}
 converges to some $g_\infty\in U$,
\begin{equation}\label{apalgorithm.thm3.eq3}
\|g_m-g_\infty\|_p\le  \frac{r_0^{m+1}}{1-r_0}  \|g_0\|_p,\ \ m\ge 0.
\end{equation}
Moreover, if $g_0=P_{U, \tilde U} S_{\Gamma, \delta} h+\tilde g$ for some $h, \tilde g\in U$, then
\begin{equation}\label{apalgorithm.thm3.eq3*}
\|g_\infty-h\|_p\le \frac{\|\tilde g\|_p}{1-r_0}.
\end{equation}
\end{Lm}



\begin{proof} 
Combining \eqref{rksadmissible.thm.eq1},
\eqref{rksadmissible.thm.pf.eq1}
 and \eqref{projection201}, we obtain
\begin{eqnarray}\label{apalgorithm.thm3.pf.eq1}
\|P_{U, \tilde U}S_{\Gamma, \delta}f-f\|_p &\le &
D_4^{-1} \sup_{g\in \tilde U, \|g\|_{p/(p-1)}\le 1}
|\langle P_{U, \tilde U}S_{\Gamma, \delta}f-f, g\rangle|\nonumber\\
& = &
D_4^{-1} \sup_{g\in \tilde U, \|g\|_{p/(p-1)}\le 1}
|\langle S_{\Gamma, \delta}f-f, g\rangle|\nonumber\\
& \le &  r_0 \|f\|_p, \ \  f\in U.
\end{eqnarray}
Observe from   \eqref{iaprks.def} that
$$g_{m+1}-g_m= (I-P_{U, \tilde U}S_{\Gamma, \delta}) (g_m-g_{m-1}),\  m\ge 1.$$
This together with \eqref{apalgorithm.thm3.pf.eq1} proves \eqref{apalgorithm.thm3.eq3}.

Now we prove \eqref{apalgorithm.thm3.eq3*}.
Taking limit in \eqref{iaprks.def} leads to the following consistence condition
\begin{equation} \label{apalgorithm.thm3.pf.eq22}
P_{U,\tilde U} S_{\Gamma, \delta} g_\infty=g_0.
\end{equation}
Replacing  $g_0$ in \eqref{apalgorithm.thm3.pf.eq22} by $P_{U, \tilde U}S_{\Gamma, \delta}h+\tilde g$ gives
\begin{equation*}
P_{U, \tilde U} S_{\Gamma, \delta} (g_\infty-h)=\tilde g.
\end{equation*}
This together with
\eqref{apalgorithm.thm3.pf.eq1} completes the proof. 
\end{proof}


\begin{proof}[Proof of Theorem \ref{galerkinrks.thm}]
Take $f\in V_{K,p}$, set $g_0=P_{U, \tilde U}S_{\Gamma, \delta} f$, and let $g_\infty\in U$ be the limit of $g_m, m\ge 0$, in the iterative algorithm
\eqref{iaprks.def}.
The existence of such a limit follows from Lemma
\ref{galerkinrks.lem}.
 Taking limit
in \eqref{iaprks.def} leads to 
\begin{equation}\label{apalgorithm.thm3.eq4}
P_{U, \tilde U}S_{\Gamma, \delta}f=P_{U, \tilde U} S_{\Gamma, \delta} g_\infty.
\end{equation}
Then
for any $g\in \tilde U$,
\begin{equation}\label{apalgorithm.thm3.eq4a}
\langle S_{\Gamma, \delta} g_\infty, g\rangle=
\langle P_{U, \tilde U} S_{\Gamma, \delta} g_\infty, g\rangle=
\langle P_{U, \tilde U} S_{\Gamma, \delta} f, g\rangle=
\langle S_{\Gamma, \delta} f, g\rangle\end{equation}
by \eqref{projection201} and  \eqref{apalgorithm.thm3.eq4}.
This proves that $g_\infty$ is a solution of Galerkin equations
\eqref{galerkinrks.thm.eq00}.

Next, we show  that $g_\infty$ is the unique solution of Galerkin equations
\eqref{galerkinrks.thm.eq00}.
Let  $h\in U$  be another  solution. 
Then
\begin{equation*}
\langle P_{U, \tilde U} S_{\Gamma, \delta} (h-g_\infty), g\rangle=
\langle S_{\Gamma, \delta}(h-g_\infty), g\rangle=0.
\end{equation*}
This together with \eqref{rksadmissible.thm.eq1} implies that
$$P_{U, \tilde U} S_{\Gamma, \delta} (h-g_\infty)=0.$$
Recall from \eqref{apalgorithm.thm3.pf.eq1} that  $P_{U, \tilde U} S_{\Gamma, \delta}$ is invertible on $U$.
Then $h=g_\infty$ and the uniqueness follows.

Observe that any $f\in U$ satisfies  Galerkin equations
\eqref{galerkinrks.thm.eq00}. This together with  \eqref{apalgorithm.thm3.eq4a}
proves that the unique solution of Galerkin equations
\eqref{galerkinrks.thm.eq00}
defines a Galerkin reconstruction.
\end{proof}

We finish this section with a remark on the iterative approximation-projection algorithm \eqref{iaprks.def}.

\begin{re}{\rm
Given $\delta>0$, a sampling set $\Gamma$ and
 probability measures $\mu_n$ supported on $I_n$, we  define
\begin{equation*}
\tilde S_{\Gamma, \delta} f(x)= \sum_{\gamma_n\in \Gamma} |I_n| f(\gamma_n) \int_{I_n} K(x, y)  d\mu_n(y),\ \
f\in V_{K,p},
\end{equation*}
where $\{I_n\subset B(\gamma, \delta), \ \gamma_n\in \Gamma\}$ is
a disjoint covering of
 $B(\Gamma, \delta)$.
 The  operator $\tilde S_{\Gamma, \delta}$ just defined becomes
the sampling operator $S_{\Gamma, \delta}$ in \eqref{rkssampling.def}
when $\mu_n$ are point measures supported on $\gamma_n$, and
the  sampling operator
\begin{equation*}
S_{\Gamma, \delta} f(x)=\sum_{\omega_n\in \Gamma}  f(\gamma_n) \int_{I_n} K(x, y) dy,\ \
f\in V_{K,p}
\end{equation*}
when $\mu_n$ are normalized Lebsegue measure supported on $I_n$.
 Following
the argument used in Theorem \ref{rksadmissible.thm} and Lemma  \ref{galerkinrks.lem}, we can show that the
approximation-projection algorithm \eqref{iaprks.def} with $S_{\Gamma, \delta}$ replaced by $\tilde S_{\Gamma, \delta}$
has exponential convergence if
$$
D_4^{-1} \big(E(U, B(\Gamma, \delta)) \|K\|_{\mathcal W}
+ \|\omega_{2\delta}(K)\|_{\mathcal W}
\big(1+ \|K\|_{\mathcal W}+
\|\omega_{2\delta}(K)\|_{\mathcal W}\big)\big)<1,$$
c.f., the second requirement \eqref{rksadmissible.thm.eq2} in Theorem \ref{rksadmissible.thm}. 
}\end{re}


\section{Sampling signals with finite rate of innovation}
\label{fri.section}

A signal with {\em finite rate of innovation} (FRI)
has  finitely many degrees of freedom per unit of time
\cite{vetterli07, mishalieldar2011, panbludragotti14, sunaicm08, sunsiam06, vetterli02}.
 Define the {\em  Wiener amalgam  space} by
 $${\mathcal W}^1:=\Big\{\phi, \  \|\phi\|_{{\mathcal W}^1}:= \sum_{k\in \ZZ} \sup_{0\le x\le 1} |\phi(x+k)|<\infty\Big\}.
 $$
 It is observed in  \cite{sunaicm08} that lots of FRI signals live in a space of the form
\begin{equation}\label{v2phi.def}
V_2(\Phi):=\Big\{ \sum_{i\in \ZZ} c_i \phi_i(\cdot-i), \ \sum_{i\in \ZZ} |c_i|^2<\infty\Big\},
\end{equation}
where the generator $\Phi:=(\phi_i)_{i\in \ZZ}$ satisfies
\begin{equation}
\label{wiener01.def}
\|\Phi\|_{{\mathcal W}^1}:=\big\|\sup_{i\in \ZZ} |\phi_i|\big\|_{{\mathcal W}^1}<\infty
\ \ {\rm and}\ \
\lim_{\delta\to 0}\big \|\sup_{i\in \ZZ}\omega_\delta(\phi_i)\big\|_{{\mathcal W}^1}=0.
\end{equation}
 In this section, we consider Galerkin reconstruction of signals  in  finite-dimensional spaces
 \begin{equation}\label{sisN.def}
 V_{2, L}(\Phi)=\Big\{\sum_{i=-L}^L c_i \phi_i(\cdot-i), \  \sum_{i=-L}^L |c_i|^2<\infty\Big\}, \ L\ge 1.\end{equation}

\subsection{Reproducing kernel spaces}
For $\Phi:=(\phi_i)_{i\in \ZZ}$ and $\tilde \Phi:=(\tilde \phi_j)_{j\in \ZZ}$ satisfying \eqref{wiener01.def}, define their
 correlation matrix by
 $$A_{\Phi, \tilde \Phi}:=
 \big(\langle \phi_i(\cdot-i), \tilde \phi_j(\cdot-j)\rangle\big)_{i,j\in \ZZ}.$$
In this subsection,  we consider when  $V_2(\Phi)$
 and $V_2(\tilde \Phi)$ in \eqref{v2phi.def} are range spaces of some idempotent integral operators  with kernels satisfying \eqref{kernel.req1} and \eqref{kernel.req2}.

 \begin{Tm} \label{frirks.tm}
 Let $\Phi$ and $\tilde \Phi$ satisfy \eqref{wiener01.def}. If
 the correlation matrix
  $A_{\Phi, \tilde \Phi}$ has  bounded inverse  on $\ell^2$, then
 $$V_2(\Phi)=V_{K,2}\quad{\rm and} \quad V_2(\tilde \Phi)=V_{K^*, 2}$$
 for some kernel $K$ satisfying \eqref{kernel.req1} and \eqref{kernel.req2}, where
 $$K^*(x,y):=K(y,x), \ x, y\in \RR.$$
 \end{Tm}

Let ${\mathcal C}_1$ contain all infinite matrices
 $A:=\big(a_{ij}\big)_{i,j\in \ZZ}$ with
 \begin{equation*}\label{gbs.def}
\|A\|_{{\mathcal C}_1}:= \sum_{k\in \ZZ}
\Big(\sup_{i-j=k}|a_{ij}|\Big)<\infty.\end{equation*}
To prove Theorem \ref{frirks.tm}, we recall Wiener's lemma for
the {\em Baskakov-Gohberg-Sj\"ostrand class}
  ${\mathcal C}_1$,
see \cite{baskakov90, gkwieot89, grochenig10,  sjostrand94, suntams07, sunca11}
and references therein.

\begin{Lm}\label{wiener.lem}
If $A\in {\mathcal C}_1$   has bounded  inverse on $\ell^2$, then its inverse $A^{-1}$ belongs to ${\mathcal C}_1$ too.
\end{Lm}

\begin{proof}[Proof of Theorem \ref{frirks.tm}]
By direct calculation, we have
$$
\|A_{\Phi, \tilde \Phi}\|_{{\mathcal C}_1}\le \|\Phi\|_{{\mathcal W}^1}
\|\tilde \Phi\|_{{\mathcal W}^1}.
$$
Thus
 the inverse of the correlation matrix $A_{\Phi, \tilde \Phi}$ belongs to the
Baskakov-Gohberg-Sj\"ostrand class
by Lemma \ref{wiener.lem}.
Write $(A_{\Phi, \tilde \Phi})^{-1}=(b_{ij})_{i,j\in \ZZ}$. One may verify that the kernel defined by
\begin{equation}\label{kernelphitildephi}
K_{\Phi, \tilde \Phi} (x,y):=\sum_{i,j\in \ZZ} \phi_i(x-i) b_{ji} \tilde \phi_j(y-j)\end{equation}
satisfies all requirements of the theorem.
%
%
\end{proof}

\subsection{Admissibility and Galerkin reconstruction}
Given a sampling set $\Gamma=\{\gamma_n\}_{n=1}^N$ ordered as $\gamma_1<\gamma_2<\cdots<\gamma_N$, define
\begin{equation}\label{samplingsis.def}
S_{\Phi, \tilde \Phi, \Gamma}f(x) :=  \sum_{n=1}^N \frac{\gamma_{n+1}-\gamma_{n-1}}{2}
f(\gamma_n) K_{\Phi, \tilde \Phi}(x, \gamma_n),\  f\in  V_2(\Phi), 
\end{equation}
and
\begin{equation}\label{agammaphi.def}
A_{\Phi, \tilde \Phi, \Gamma}:=
\Big(\sum_{n=1}^N \frac{\gamma_{n+1}-\gamma_{n-1}}{2} \phi_i(\gamma_n-i) \tilde \phi_j(\gamma_n-j)\Big)_{-L\le i, j\le L}, \ L\ge 1,\end{equation}
 where
$\gamma_0=\gamma_1, \gamma_{N+1}=\gamma_N$, and the kernel
$K_{\Phi, \tilde \Phi}$ is given in  \eqref{kernelphitildephi}.
%
%
%
%
%
%
%
In this subsection,  we  investigate  admissibility  of  the  operator $S_{\Phi, \tilde \Phi, \Gamma}$
and its corresponding Galerkin reconstruction,
 c.f. Corollary
\ref{Rexistence.cor}, and Theorems \ref{rksadmissible.thm} and \ref{galerkinrks.thm}.

\begin{Tm}\label{v2phiadmissible.tm}  Let $\Phi$ and $\tilde \Phi$ satisfy \eqref{wiener01.def}. Assume that
the correlation matrix $A_{\Phi, \tilde \Phi}$ has bounded inverse on $\ell^2$.
Then the following statements are equivalent:

  \begin{itemize}

\item [{(i)}]  The $L\times L$ matrix
$A_{\Phi, \tilde \Phi, \Gamma}$ in \eqref{agammaphi.def}
is nonsingular.

\item [{(ii)}] $S_{\Phi, \tilde \Phi, \Gamma}$ is admissible for the pair $(V_{2, L}(\Phi), V_{2, L}(\tilde \Phi))$.

\item[{(iii)}] For any $f\in V_{2}(\Phi)$,
Galerkin equations
\begin{equation}\label{v2phigalerkin.def}
\langle S_{\Phi, \tilde \Phi, \Gamma} h, g\rangle=\langle S_{\Phi, \tilde \Phi, \Gamma} f, g\rangle, \ g\in V_{2, L}(\tilde \Phi)
\end{equation}
have a unique solution  $h$ in $V_{2, L}(\Phi)$.

\item[{(iv)}] For any $g\in V_2(\tilde\Phi)$, dual Galerkin equations

\begin{equation*}\label{v2phigalerkin.def2}
\langle S_{\Phi, \tilde \Phi, \Gamma}f, \tilde h\rangle=\langle S_{\Phi, \tilde \Phi, \Gamma}f, g\rangle, \ f\in V_{2, L}(\Phi)
\end{equation*}
have a unique solution  $\tilde h$ in $V_{2, L}(\tilde\Phi)$.

\end{itemize}

 \end{Tm}

   \begin{proof} 
For $h=\sum_{i=-L}^L c_i \phi_i(\cdot-i)\in V_{2, L}(\Phi)$ and
$g=\sum_{j=-L}^L d_j \tilde \phi_j(\cdot-j)\in V_{2, L}(\tilde \Phi)$, we obtain
\begin{eqnarray}\label{v2phigalerkin.tm.pf.eq1}
\langle S_{\Phi, \tilde \Phi, \Gamma} h, g\rangle & = &
\sum_{i,j=-L}^L  \Big(\sum_{n=1}^N \frac{\gamma_{n+1}-\gamma_{n-1}}{2} \phi_i(\gamma_n-i) \langle K_{\Phi, \tilde \Phi}(t, \gamma_n),
\tilde \phi_j(t-j)\rangle\Big) c_i d_j\nonumber\\
& = &
\sum_{i,j=-L}^L  \Big(\sum_{n=1}^N \frac{\gamma_{n+1}-\gamma_{n-1}}{2} \phi_i(\gamma_n-i)
\tilde \phi_j(\gamma_n-j)\Big) c_i d_j\nonumber\\
&= & c^T A_{\Phi, \tilde \Phi, \Gamma} d,
\end{eqnarray}
where $c=(c_i)_{-L\le i\le L}$ and $d=(d_j)_{-L\le j\le L}$.
By the invertibility assumption on $A_{\Phi, \tilde \Phi}$,
 $\{\phi_i(\cdot-i), -L\le i\le L\}$ and $\{\tilde \phi_i(\cdot-i), -L\le i\le L\}$
  are Riesz bases of $V_{2, L}(\Phi)$ and $V_{2, L}(\tilde \Phi)$ respectively.
  This together with \eqref{v2phigalerkin.tm.pf.eq1} proves the desired equivalent statements.
%
\end{proof}

%
%
%
%
%
\subsection{Oblique Projection and iterative approximation-projection algorithm}
In this subsection,  we first discuss  existence and uniqueness of oblique projection for the pair $(V_{2, L}(\Phi), V_{2,L}(\tilde \Phi))$.

\begin{Tm}  Let $L\ge 1$, and let $\Phi$ and $\tilde \Phi$ satisfy \eqref{wiener01.def}. Assume that
the correlation matrix $A_{\Phi, \tilde \Phi}$ has bounded inverse on $\ell^2$.
Then
the principal submatrix
\begin{equation}\label{aphitildeL.def}
A_{\Phi, \tilde\Phi, L}:=\big(\langle \phi_i(\cdot-i), \tilde \phi_j(\cdot-j)\rangle\big)_{-L\le i,j\le L}\end{equation}
of the correlation matrix $A_{\Phi, \tilde \Phi}$ is nonsingular if and only if
there exists a unique oblique projection for the pair $(V_{2, L}(\Phi), V_{2, L}(\tilde \Phi))$.
Moreover, the oblique projection could be defined by
\begin{equation}\label{projectionphitildephi.def}
P_{\Phi, \tilde \Phi, L}f:=\sum_{-L\le i,j\le L}  \langle f, \tilde \phi_i(\cdot-i)  \rangle \tilde b_{ij} \phi_j(\cdot-j), \ f\in V_2(\Phi),\end{equation}
where $(A_{\Phi, \tilde\Phi, L})^{-1}=(\tilde b_{ij})_{-L\le i,j\le L}$.
\end{Tm}

\begin{proof}
The sufficiency is obvious.
Now we prove the necessity. Suppose, to the contrary, that $A_{\Phi, \tilde \Phi, L}$ in \eqref{aphitildeL.def} is singular.
Take  a nonzero vector $e=(e_i)_{-L\le i\le L}$  in the null space $N((A_{\Phi, \tilde \Phi, L})^T)$
and  a nonzero linear functional ${\mathcal J}$  on $V_2(\Phi)$ such that ${\mathcal J}(h)=0$ for all $h\in V_{2, L}(\Phi)$. Define
$$Q(f):={\mathcal J}(f) \sum_{-L\le i\le L} e_i \phi_i(\cdot-i), \ f\in V_2(\Phi).$$
Then  $Q$ is a nonzero linear operator from $V_2(\Phi)$ to $V_{2, L}(\Phi)$, 
$$Qh=0,\quad  h\in V_{2, L}(\Phi)$$
 and
$$ \langle Qf, g\rangle= {\mathcal J}(f) \sum_{-L\le i,j\le L}  e_i \langle \phi_i(\cdot-i), \tilde \phi_j(\cdot-j)\rangle d_j=0,
$$
where $g=\sum_{-L\le j\le L} d_j \tilde \phi_j(\cdot-j)\in V_{2, L}(\tilde \Phi)$.
This contradicts to the uniqueness of oblique projections. 
\end{proof}

In this subsection, we then examine  exponential convergence of an iterative algorithm
for the recovery of  signals with finite rate of innovation.
 Replacing $P_{U, \tilde U}$ and $S_{\Gamma, \delta}$ in
the iterative algorithm \eqref{iaprks.def}
by $P_{\Phi, \tilde \Phi, L}$ and $S_{\Phi, \tilde \Phi, \Gamma}$ respectively,
it becomes
\begin{equation} \label{apalgorithm.thm3.eq2bb}
g_{m+1}=g_m-\sum \limits_{n=1}^N \sum \limits_{ i,j=-L}^{L}
\frac{\gamma_{n+1}-\gamma_{n-1}}{2} g_m(\gamma_n) \tilde \phi_i(\gamma_n-i) \tilde b_{ij} \phi_j(\cdot-j)+g_0, \ m\ge 0,\end{equation}
with $g_0\in V_{2, L}(\Phi)$.

\begin{Tm}\label{v2phialgorithm.thm}
 Let $\Phi$ and $\tilde \Phi$ satisfy \eqref{wiener01.def}.
 Assume that
$A_{\Phi, \tilde\Phi, L}$ is nonsingular.
If
\begin{equation}\label{v2phialgorithm.thm.eq1}
\|A_{\Phi, \tilde \Phi, \Gamma}  (A_{\Phi, \tilde \Phi, L})^{-1}-I\|<1,\end{equation}
then the iterative algorithm
\eqref{apalgorithm.thm3.eq2bb}
has exponential convergence. Moreover, it recovers the original signal
$h\in V_{2, L}(\Phi)$  when
$$g_0=\sum \limits_{n=1}^N \sum \limits_{ i,j=-L}^{L}
\frac{\gamma_{n+1}-\gamma_{n-1}}{2} h(\gamma_n) \tilde \phi_i(\gamma_n-i) \tilde b_{ij} \phi_j(\cdot-j).
$$
\end{Tm}

\begin{proof}
Write $g_m=\sum_{-L\le i\le L} c_m(i)\phi_i(\cdot-i)$ and set
$c_m=(c_m(i))_{-L\le i\le L}$.
Then we can reformulate the iterative algorithm \eqref{apalgorithm.thm3.eq2bb}
as
\begin{equation*}\label{v2phialgorithm.thm.pf.eq1}
c_{m+1}^T= c_m^T- c_m^T A_{\Phi, \tilde \Phi, \Gamma}  (A_{\Phi, \tilde \Phi, L})^{-1}+c_0^T, \ m\ge 0.
\end{equation*}
This together with \eqref{v2phialgorithm.thm.eq1} proves the desired conclusions.
\end{proof}


\section{Numerical Simulation}
\label{simulation.section}

In this section, we present several  examples to illustrate our  Galerkin reconstruction
of signals with finite rate of innovation.

\smallskip

Let $\Theta:=\{\theta_i\}$ be either
$\Theta_O:=\{0\}$ (the identical zero set),  or $\Theta_I$ with $\theta_i$ being randomly selected in $[-0.2, 0.2]$.
Set
$$\Phi_0=\{\phi_0(\cdot-\theta_i)\}_{i\in \ZZ},$$
 where
 the generating function $\phi_0$ is either
(i)  the sinc function
${\rm sinc}(t):=\frac{\sin\pi t}{\pi t}$, or
(ii) the Gaussian function ${\rm gauss}(t):=\exp(-3t^2/2)$, or
(iii)  the cubic $B$-spline  ${\rm spline}(t)$, 
see Figure \ref{originalsinc.figure} for examples of signals in $V_2(\Phi_0)$.
%
\begin{figure}[hbt]
\centering
\begin{tabular}{cc}
 \includegraphics[width=60mm]{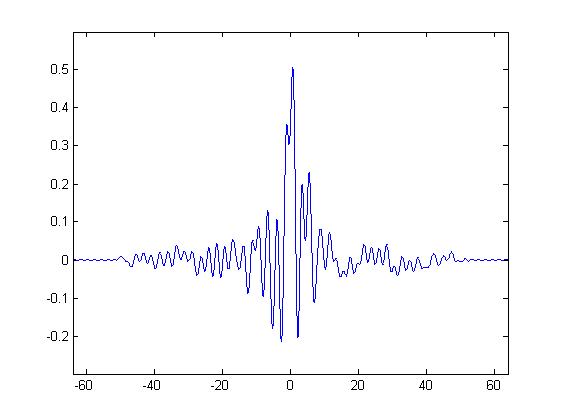}&
  \includegraphics[width=60mm]{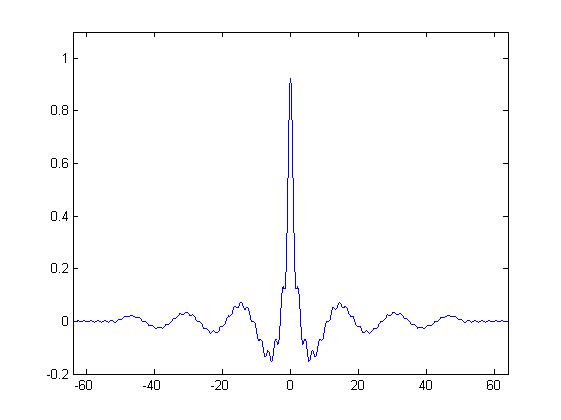}
\\
 \includegraphics[width=60mm]{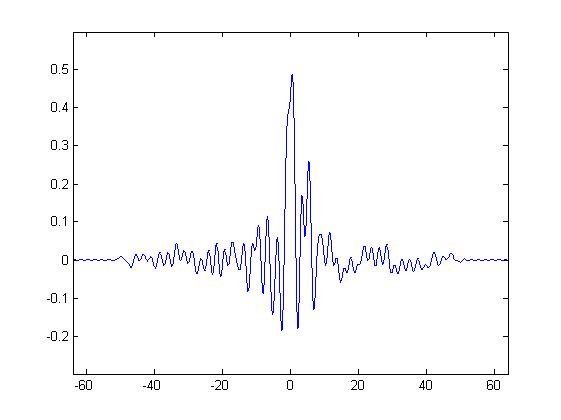}
 &
 \includegraphics[width=60mm]{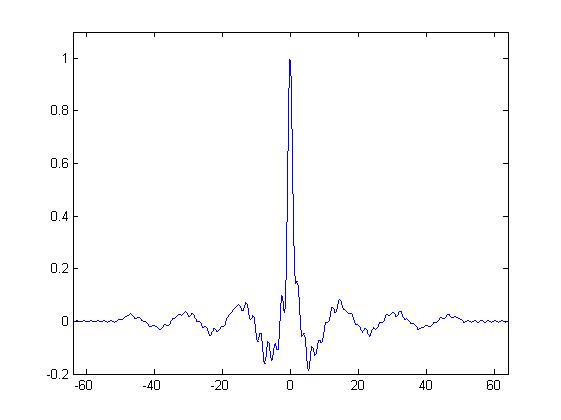}
  \\
  \end{tabular}
\caption{\small Above are bandlimited signals $x({\rm sinc},  0)=\sum_i \alpha_i {\rm sinc} (t-i)$ with
    $(1+|i|)\alpha_i\in [-1, 1]$  randomly selected (left), and
$x({\rm sinc}, 1)=\sum_i \beta_i {\rm sinc} (t-i)$  with  $\beta_i=(1+|i|)^{-1} \cos(\pi i/8)$ (right).
Below are signals
$x({\rm sinc}, 2)=\sum_i \alpha_i {\rm sinc} (t-i-\theta_i)$  (left) and
 $x({\rm sinc}, 3)=\sum_i \beta_i {\rm sinc} (t-i-\theta_i)$
with $\theta_i\in [-0.2, 0.2]$ randomly selected (right).
 } \label{originalsinc.figure}
\end{figure}
In our numerical simulations, reconstructed signals   live  in the space 
$$V_{2,L}(\Phi_0)=\Big\{\sum_{i=-L}^L c_i \phi_0(t-i-\theta_i): \sum_{i=-L}^L |c_i|^2<\infty\Big\}, \ L \ge 1,$$
 and
sampling schemes are

 \begin{itemize}
\item
Nonuniform sampling on $\Gamma_N:=\{\gamma_k, |k|\le L+2\}$,
 where $\gamma_{-L-3}=-L-2$ and
  $\gamma_{k}-\gamma_{k-1}\in [0.9, 1.1], |k|\le L+2$, are randomly selected.

   \item Jittered sampling
on
$\Gamma_J:=\{\gamma_k:=k+\delta_k, |k|\le L+2\}$,
where $\delta_k\in [-0.1, 0.1]$  are randomly selected.

\item Adaptive sampling on $\Gamma_C:=\{\gamma_k\in [-L-2, L+2]\}$ of a bounded signal $x\in V_2(\Phi)$  via crossing time encoding machine (C-TEM), where
 $x(t)\ne \|x\|_\infty \sin(\pi t)$ for all $t\in [-L-2, L+2]$ except $t=\gamma_k$ for some $k$, see Figure \ref{ctem.fig}  \cite{feichtinger2012, vetterli2014, Lazar2004}.
 \end{itemize}

 \begin{figure}[hbt]
\centering
\begin{tabular}{l}
 \includegraphics[width=128mm, height=35mm]{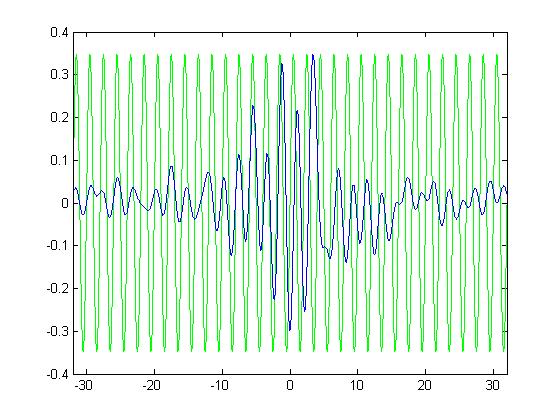}\\
 \includegraphics[width=128mm, height=30mm]{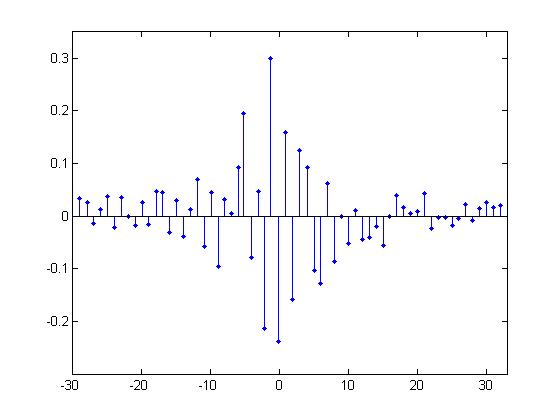}
  \end{tabular}
\caption{\small Above is the signal $x({\rm sinc}, 0)$ in Figure \ref{originalsinc.figure}
and the crossing signal $\|x({\rm sinc},0)\|_\infty \sin \pi t$ on $[-L-2, L+2]$,  and below
is  the sampling data of $x({\rm sinc}, 0)$ on  the sampling set $\Gamma_C\subset [-L-2, L+2]$, where $L=30$.
 } \label{ctem.fig}
\end{figure}

 To reconstruct signals via Galerkin method, we take
 $$\tilde \Phi_0=\{\tilde \phi_0\}\quad {\rm  with} \quad \tilde \phi_0=\chi_{[-1/2, 1/2)}.$$
Then  the equation \eqref{v2phigalerkin.def} to determine
the  Galerkin reconstruction
$$G_{\Phi_0, \tilde \Phi_0, \Gamma}f:=\sum_{i=-L}^L c_i \phi_0(\cdot-i-\theta_i)\in V_{2, L}(\Phi_0)$$
  can be reformulated as
follows:
\begin{eqnarray}\label{v2phi0galerkin}
& & \sum_{i=-L}^L
\Big(\sum_{n=1}^N \frac{\gamma_{n+1}-\gamma_{n-1} }{2}
\phi_0(\gamma_n-i-\theta_i)\tilde \phi_0(\gamma_n-j)\Big) c_i\nonumber\\
& = & \sum_{n=1}^N \frac{\gamma_{n+1}-\gamma_{n-1} }{2} f(\gamma_n) \tilde \phi_0(\gamma_n-j),
-L\le j\le L,
\end{eqnarray}
where  $f\in V_2(\Phi_0)$
and $\Gamma:=\{\gamma_n\}_{n=1}^N$ is either the nonuniform sampling set $\Gamma_N$, or
the jittered sampling set $\Gamma_J$, or the adaptive C-TEM sampling set $\Gamma_C$.
Considering the bandlimited signal $x({\rm sinc}, 0)$
  described in Figure \ref{originalsinc.figure}, we present some numerical results for
  its pre-reconstruction in $V_2(\Phi_0)$ and  Galerkin reconstruction  in $V_{2, L}(\Phi_0)$  in Figure \ref{prereconstruction.fig}. We see that a pre-reconstruction may provide a reasonable
approximation, while
a Galerkin reconstruction could recover the original signal
almost perfectly in the sampling interval. 
%
%
\begin{figure}[hbt]
\centering
\begin{tabular}{cc}
 \includegraphics[width=62mm, height=38mm]{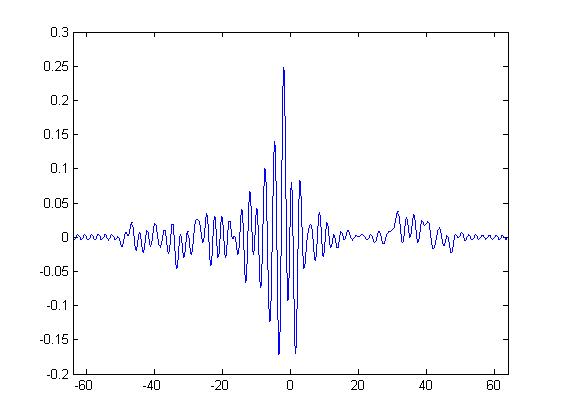}&
   \includegraphics[width=62mm, height=38mm]{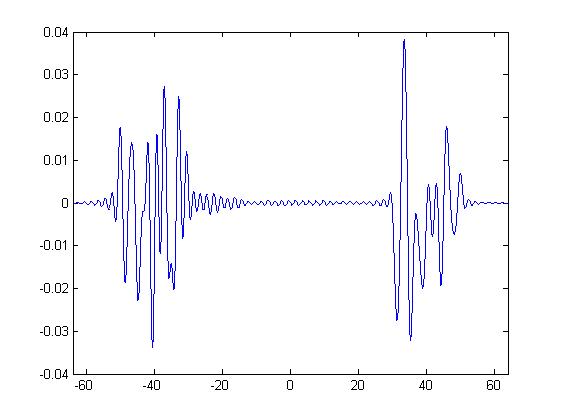}
   \\
  \includegraphics[width=62mm,  height=38mm]{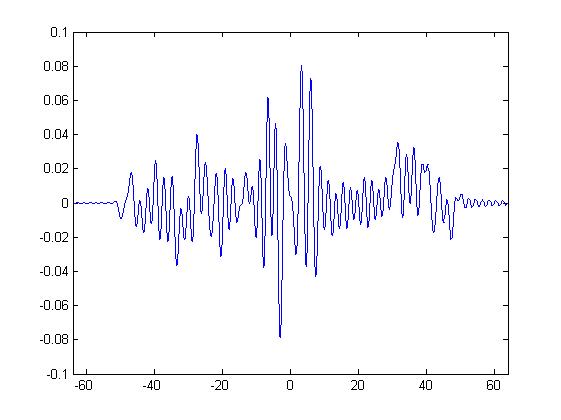}
 &
  \includegraphics[width=62mm, height=38mm]{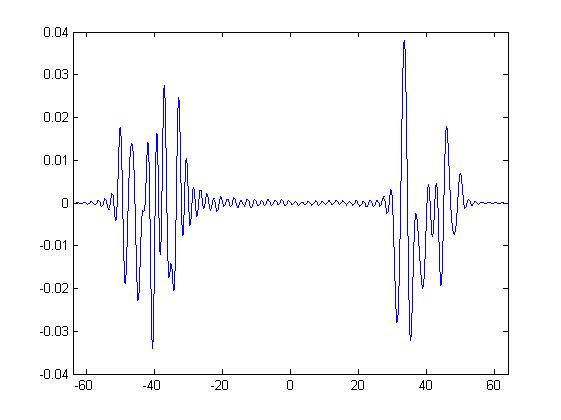}\\
       \includegraphics[width=62mm, height=38mm]{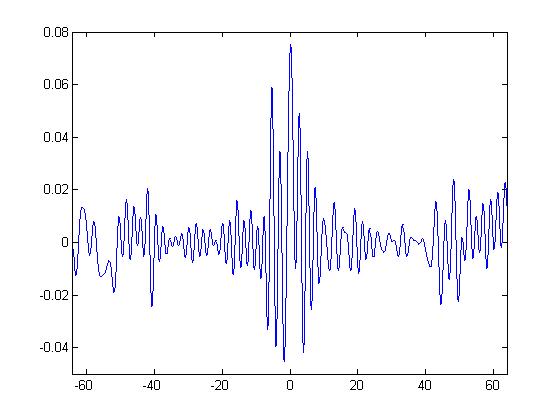}
 &
  \includegraphics[width=62mm, height=38mm]{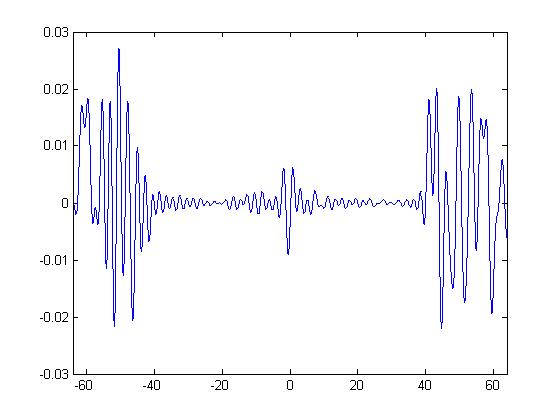}
  \end{tabular}
\caption{\small On the top left is the difference between the  signal $x({\rm sinc}, 0)$ in Figure \ref{originalsinc.figure}
and its pre-reconstructed signal
$S_{ \Phi_0, \tilde \Phi_0,\Gamma_N} x({\rm sinc},  0)$,
while on the top right is the difference between  $x({\rm sinc}, 0)$ and its
Galerkin reconstruction  $G_{ \Phi_0, \tilde \Phi_0,\Gamma_N} x({\rm sinc},  0)$.
The middle are differences $x({\rm sinc}, 0)-S_{ \Phi_0, \tilde \Phi_0, \Gamma_J} x({\rm sinc},  0)$ (left)
and $x({\rm sinc}, 0)-G_{ \Phi_0, \tilde \Phi_0, \Gamma_J}x({\rm sinc}, 0)$ (right)
associated with jittered sampling.
Listed below are differences
$x({\rm sinc}, 0)-S_{ \Phi_0, \tilde \Phi_0, \Gamma_C} x({\rm sinc},  0)$ (left)
and $x({\rm sinc}, 0)-G_{ \Phi_0, \tilde \Phi_0, \Gamma_C}x({\rm sinc}, 0)$ (right)
associated with adaptive C-TEM sampling.
 } \label{prereconstruction.fig}
\end{figure}

\smallskip

For $\Phi_0=\{\phi_0(\cdot-\theta_i)\}$,
let signals $x(\phi_0, l)\in V_2(\Phi_0),  0\le l\le 3$, be as
$x({\rm sinc}, l)$
in Figure \ref{originalsinc.figure} with the sinc function replaced by the function $\phi_0$.
In Figure \ref{galerkinreconstruction.fig}, we illustrate
their best approximation  in $V_{2,L}(\Phi_0)$
and  solutions of the Galerkin system \eqref{v2phi0galerkin}
with $f$ replaced by $x(\phi_0, l), 0\le l\le 3$, respectively.
We observe that
given a signal in $V_2(\Phi_0)$, its Galerkin reconstruction in $V_{2, L}(\Phi_0)$
could almost match its best approximation
in $V_{2, L}(\Phi_0)$, except near the boundary of the sampling interval.
The boundary effect is viewable especially when 
 $\phi_0$
has slow decay at infinity.
\begin{figure}[hbt]
\centering
\begin{tabular}{cc}
 \includegraphics[width=62mm]{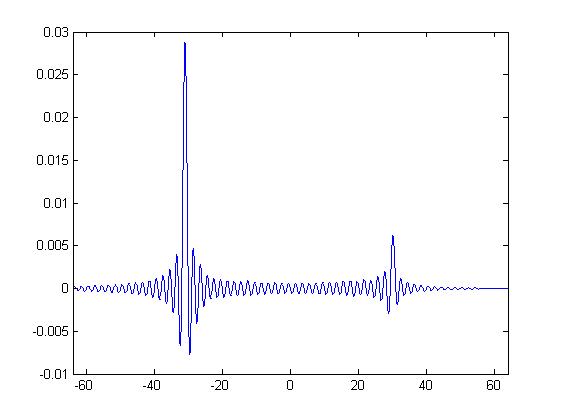}&
  \includegraphics[width=62mm]{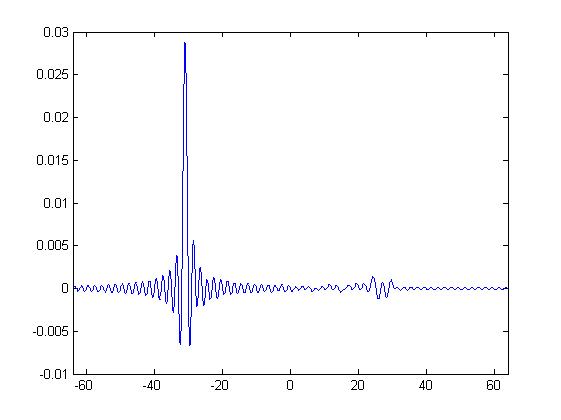}
  \\
   \includegraphics[width=62mm]{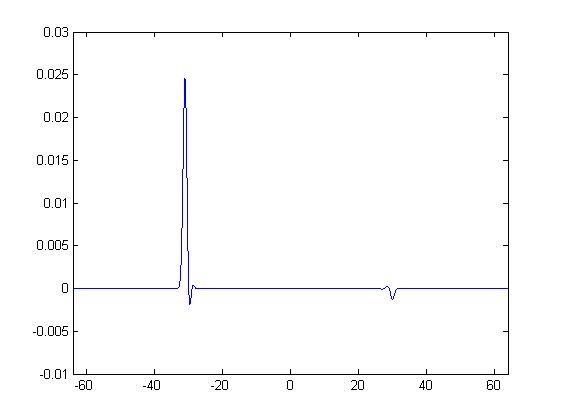}&
   \includegraphics[width=62mm]{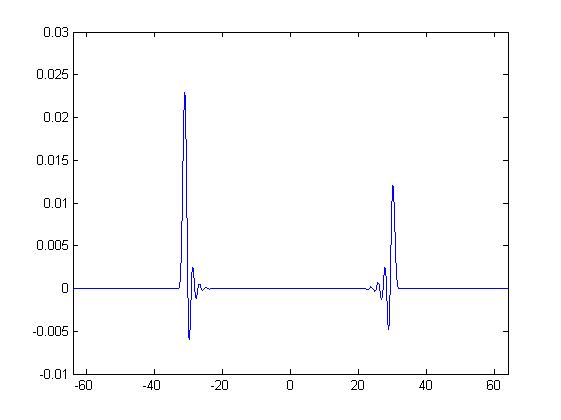}
\\
  \end{tabular}
\caption{\small Listed  are
 differences between best approximations of
 signals $x(\phi_0, 0)$ in $V_{2, 30}(\Phi_0)$ and
 their Galerkin reconstructions associated with
 operators $S_{\Phi_0, \tilde \Phi_0, \Gamma}$,
  where on the above, $\phi_0={\rm sinc}$,
  $\Gamma=\Gamma_N$(left)  and $\Gamma=\Gamma_J$ (right), while on the bottom
  $\Gamma=\Gamma_N$,
  $\phi_0={\rm gauss}$ (left)  and $\phi_0={\rm spline}$ (right).
%
%
 } \label{galerkinreconstruction.fig}
\end{figure}

\smallskip

Given signals $x(\phi_0, l), 0\le l\le 3$,
let  $y_L(\phi_0, l)$
 be their best approximators  in $V_{2,L}(\Phi_0)$,
 and
 denote by
 $$e(\phi_0, l)=\|x(\phi_0, l)-y_L(\phi_0, l)\|$$
 their best approximation error
 in $V_{2, L}(\Phi_0)$.
For $\Gamma=\Gamma_N$ or $\Gamma_J$ or $\Gamma_C$,
  set
$$\epsilon_\Gamma(\phi_0, l)=\|z_L(\Gamma, \phi_0, l)-y_L(\phi_0, l)\|,$$
 where $z_L(\Gamma, \phi_0, l)$
 is obtained from solving Galerkin system
 \eqref{v2phi0galerkin} with $f$ replaced by $x(\phi_0, l)$.
For signals $x(\phi_0, l), 0\le l\le 3$,
 and sampling sets $\Gamma=\Gamma_N, \Gamma_J$ and $\Gamma_C$, Galerkin reconstruction
 \eqref{v2phi0galerkin}
 provides quasi-optimal approximation in  $V_{2, L}(\Phi_0)$,
and
the  quasi-optimal constant
in  Theorem \ref{nsqs.thm}
 is well behaved,
%
%
  $$\frac{\|z_L(\Gamma, \phi_0, l)-x(\phi_0, l)\|}{\|y_L(\phi_0, l)-x(\phi_0, l)\|}
 \le 1+\frac{ \epsilon_\Gamma(\phi_0, l)}{e(\phi_0, l)}\le \frac{3}{2},$$
 see
 Table  \ref{galerkin.tab} for numerical results with abbrievated notations.
  \begin{table}[h]
\caption{Quasi-optimality of Galerkin reconstructions for bandlimited/Gauss/spline signals
}
\begin{tabular}{ccccccc}
\hline\hline
L 
 &  10 
& 15 
 &  20 
 & 25 & 30
 \\
\hline
$ e({\rm sinc}, 0) $ &
 0.2176  &  0.1711  &   0.1388  &   0.1166 &   0.1024
   \\
$\epsilon_N({\rm sinc}, 0)$ &
0.0795  &  0.0668  &  0.0197  &  0.0201  &  0.0294\\
$\epsilon_J({\rm sinc}, 0)$ &
0.0770  & 0.0668  &  0.0201  & 0.0214  & 0.0290\\
$\epsilon_C({\rm sinc}, 0)$ &
0.0789 &  0.0715  &  0.0239  &  0.0263 &  0.0325\\
\hline
$ e({\rm sinc}, 1) $ &
0.2600 &  0.2124 & 0.1816 & 0.1457  & 0.1303\\
$ \epsilon_N({\rm sinc}, 1) $ &
0.0344 & 0.0809 &  0.0370  & 0.0294  &  0.0431\\
$ \epsilon_J({\rm sinc}, 1) $ &
0.0353  & 0.0806  & 0.0372  & 0.0301  & 0.0433\\
$ \epsilon_C({\rm sinc}, 1) $ &
0.0363  & 0.0831  & 0.0379  & 0.0319  & 0.0442\\
\hline
$ e({\rm sinc}, 2) $ &
 0.2095  &  0.1703 & 0.1365 &  0.1167 &  0.1007
     \\
$ \epsilon_N({\rm sinc}, 2) $ &
0.0619 &  0.0618  & 0.0256 & 0.0163  & 0.0281\\
$ \epsilon_J({\rm sinc}, 2) $ &
0.0596  &  0.0618  & 0.0260  &  0.0177   & 0.0275\\
$ \epsilon_C({\rm sinc}, 2) $ &
0.0608  &  0.0664  & 0.0284  & 0.0226   & 0.0308\\
\hline
$ e({\rm sinc}, 3) $ &
0.2655 & 0.2180 & 0.1863  & 0.1477  & 0.1322
 \\
$ \epsilon_N({\rm sinc}, 3) $ &
 0.0461  &  0.0810  &  0.0374  & 0.0258  & 0.0406\\
$ \epsilon_J({\rm sinc}, 3) $ &
0.0446 & 0.0809 & 0.0375  & 0.0265  & 0.0401\\
$ \epsilon_C({\rm sinc}, 3) $ &
0.0474 & 0.0837 & 0.0392  & 0.0298  & 0.0418\\
\hline
$ e({\rm gauss}, 0) $ &
 0.2055 & 0.1682  & 0.1398  & 0.1250 & 0.1086
   \\
$\epsilon_N({\rm gauss}, 0)$ &
0.0437 & 0.0515  & 0.0270 & 0.0158 & 0.0093
\\
$\epsilon_J({\rm gauss}, 0)$ &
0.0439  &  0.0523 & 0.0259  & 0.0160  &  0.0096\\
$\epsilon_C({\rm gauss}, 0)$ &
0.0433  & 0.0527 & 0.0270  & 0.0181  & 0.0108\\
\hline
$ e({\rm spline}, 0) $ &
0.1482 &  0.1325  &  0.1110   & 0.0924   & 0.0664   \\
$\epsilon_N({\rm spline}, 0)$ &
 0.0405  &  0.0298  &  0.0204 & 0.0266  & 0.0176
\\
$\epsilon_J({\rm spline}, 0)$ &
 0.0403  & 0.0299  &  0.0204 &  0.0281  &  0.0184\\
 $\epsilon_C({\rm spline}, 0)$ &
 0.0407 & 0.0292  &  0.0209 &  0.0279  & 0.0181\\
\hline \hline
\end{tabular}
\label{galerkin.tab}
\end{table}

\smallskip

Numerical stability of Galerkin reconstruction \eqref{v2phi0galerkin}
could be reflected by the condition number
${\rm cond}_{\Gamma, \Theta}(\phi_0)$ of the square matrix
$$A_{\Phi_0, \tilde \Phi_0, \Gamma}=\Big(\sum_{n=1}^N \frac{\gamma_{n+1}-\gamma_{n-1} }{2}
\phi_0(\gamma_n-i-\theta_i)\tilde \phi_0(\gamma_n-j)\Big)_{-L\le i, j\le L}. $$
 Some numerical results of
 condition numbers ${\rm cond}_{\Gamma, \Theta}(\phi_0)$
  with $\Gamma=\Gamma_N$ or $\Gamma_J$, and
$\Theta=\Theta_O$ or $\Theta_I$, are presented in  Table  \ref{galerkinstability.tab} with abbreviated notations.
\begin{table}
\caption{Stability of Galerkin reconstructions for nonuniform/jittered sampling
}
\begin{tabular}{ccccccc}
\hline\hline
L 
 &  10 
& 15 
 &  20 
 & 25 & 30
 \\
\hline
${\rm cond}_{N,O}({\rm sinc})$ &
    1.2059 &  1.2367 &  1.3458 &  1.4273 &  1.2904
    \\
${\rm cond}_{N,I}({\rm sinc})$ &
 1.9190 &  1.8946  &  1.9828 &  2.0635  & 2.0421
\\
\hline
${\rm cond}_{N, O}({\rm gauss})$ &
3.0162 &  2.7000 & 2.7908 & 3.3314 &  2.8362\\
${\rm cond}_{N,I}({\rm gauss})$ &
  3.2850 & 3.1447 & 3.1421 & 4.0283 &  3.4391\\
    \hline
${\rm cond}_{N,O}({\rm spline})$ &
3.7677 &  3.7534 &  3.0534 &  3.1400  &  4.1708\\
${\rm cond}_{N,I}({\rm spline})$ &
  4.4768  & 5.2417 &  3.3507 &  3.5354  &  5.0292\\
    \hline\hline
${\rm cond}_{J,O}({\rm sinc})$ &
    1.3737 &  1.4164  &  1.4105 &  1.4149 &  1.3763\\
${\rm cond}_{J,I}({\rm sinc})$ &
   1.9723 &  1.9351 &  2.3328 & 2.2037 &  2.1744\\
    \hline
${\rm cond}_{J,O}({\rm gauss})$ &
2.7066  & 2.7074  &  2.6936 & 2.6957 & 2.7190\\
${\rm cond}_{J,I}({\rm gauss})$ &
3.0847  &  3.1591 & 3.0696 & 3.0197  & 3.0878
\\ \hline
${\rm cond}_{J,O}({\rm spline})$ &
 3.1052 &  3.2109 &  3.2218 &  3.3257  &  3.2331
    \\
${\rm cond}_{J,I}({\rm spline})$ &
3.5570 &   3.7388  &  3.7140 &  3.9172  & 4.1830
\\
\hline \hline
\end{tabular}
\label{galerkinstability.tab}
\end{table}
 For the robust (sub-)Galerkin reconstruction,
 the generating function $\tilde \phi_0$ of the test space $V_{2, L}(\tilde \Phi_0)$
 should be so chosen that the corresponding  matrice $A_{\Phi_0, \tilde \Phi_0, \Gamma}$ is well-conditioned,
  c.f. Theorem \ref{nsqs.thm}.

\medskip
We conclude this sections with two more remarks.

\begin{re} {\rm The iterative approximation-projection algorithm \eqref{apalgorithm.thm3.eq2bb}
 could have better performance on solving
 Galerkin equations \eqref{v2phi0galerkin}, especially while matrices $A_{\Phi_0, \tilde \Phi_0, \Gamma}$
 have large condition number, which is the case when the sampling set $\Gamma$ and/or the shifting set
 $\Theta$ are not chosen appropriately.
 }\end{re}

 \begin{re} {\rm For the admissibility of
 the pre-reconstruction operator $S_{\Gamma, \delta}$, 
the test space $\tilde U$ must have its dimension larger than or equal to the one of
  the reconstruction space $U$.
For $U=V_{2, L}(\Phi_0)$ and $\tilde U=V_{2, \tilde L}(\tilde \Phi_0)$ with $\tilde L\ge L$,
least square solutions of the linear system \eqref{v2phi0galerkin}\ with $-L \le j \le L$ replaced by $-\tilde L\le j\le \tilde L$
defines a  sub-Galerkin reconstruction
$\sum_{i=-L}^L c_i \phi_0(\cdot-i-\theta_i)\in V_{2, L}(\Phi_0)$ by Corollary  \ref{hilbertreconstruction.cor},
where  $f\in V_2(\Phi_0)$
and $\Gamma:=\Gamma_N, \Gamma_J, \Gamma_C$.
Our numerical simulations show that the above sub-Galerkin reconstructions  for different  $\tilde L\ge L$
have comparable approximation errors.
} \end{re}

{\bf Acknowledgement}  The authors  thank Dr. B. Adcock for his comments and suggestions.
The project is partially supported by the National Natural Science Foundation of China (Nos. 11201094 and 11161014),
Guangxi Natural Science Foundation (2014GXNSFBA118012),
Program for Innovative Research Team of Guilin University of Electronic Technology, and
the National Science  Foundation (DMS-1109063 and DMS-1412413).

\end{document}